\documentclass[sigconf,nonacm]{aamas}

\fancypagestyle{arxiv}{%
  \fancyhf{}%
  \fancyfoot[C]{\thepage}}
\usepackage{balance}

\usepackage{newfloat}
\usepackage{listings}
\DeclareCaptionStyle{ruled}{labelfont=normalfont,labelsep=colon,strut=off}
\floatstyle{ruled}
\newfloat{listing}{tb}{lst}{}
\floatname{listing}{Listing}

\newcommand{\cmark}{\ding{51}}
\newcommand{\xmark}{\ding{55}}

\usepackage{multirow}
\usepackage{enumitem}
\usepackage{makecell}
\usepackage{xcolor}
\usepackage{colortbl}
\definecolor{ours}{RGB}{235,235,235}
\usepackage{etoolbox}
\preto\cmidrule{\noalign{\global\rownum=0\relax}}

\newtheorem{proposition}{Proposition}

\usepackage{mdframed}

\newcount\Comments
\newcommand{\kibitz}[2]{\ifnum\Comments=1\textcolor{#1}{#2}\fi}

\setcopyright{none}
\renewcommand\footnotetextcopyrightpermission[1]{}

\title[AAMAS-2027 Formatting Instructions]{AI Alignment for Digital Democracy: Fairer Voters Representation against LLM Biases}

\author{Yi Qi}
\affiliation{
  \institution{University of Leeds}
  \city{Leeds}
  \country{United Kingdom}}
\email{Y.Qi@leeds.ac.uk}

\author{Srijoni Majumdar}
\affiliation{
  \institution{University of Leeds}
  \city{Leeds}
  \country{United Kingdom}}
\email{S.Majumdar@leeds.ac.uk}

\author{Evangelos Pournaras}
\affiliation{
  \institution{University of Leeds}
  \city{Leeds}
  \country{United Kingdom}}
\email{E.Pournaras@leeds.ac.uk}

\begin{abstract}
AI representation is emerging in digital democracy, assisting in voter emulation and decision-making. Large language models (LLMs) behave differently depending on the model, the prompt context, and the voter information they receive, representing some groups more accurately than others and potentially making the emulation unfair to certain groups. We investigate demographic biases in LLM-emulated electorates: how accurately LLMs represent individual votes, whether some voter groups are represented more accurately than others, and whether AI alignment can reduce these disparities without sacrificing overall accuracy. Using three U.S. presidential elections, we compare a closed-source and an open-source LLM across voters' demographic features and test AI alignment strategies toward fairer voter representation. 
Removing gender from the voter profile reduces the largest group error of the the models by 22–42\%, and removing education reduces it by 64\%, whereas removing self-reported race increases the mean group error by over 40\%, all without lowering accuracy.
Common alignment strategies, such as repeated queries or fairness instructions, do not reduce these errors and can double them.
\end{abstract}


\newcommand{\BibTeX}{\rm B\kern-.05em{\sc i\kern-.025em b}\kern-.08em\TeX}

\begin{document}

\maketitle
\thispagestyle{arxiv}
\pagestyle{arxiv}
\section{Introduction}
\begin{figure}[htbp]
    \centering
    \includegraphics[width=0.9\linewidth]{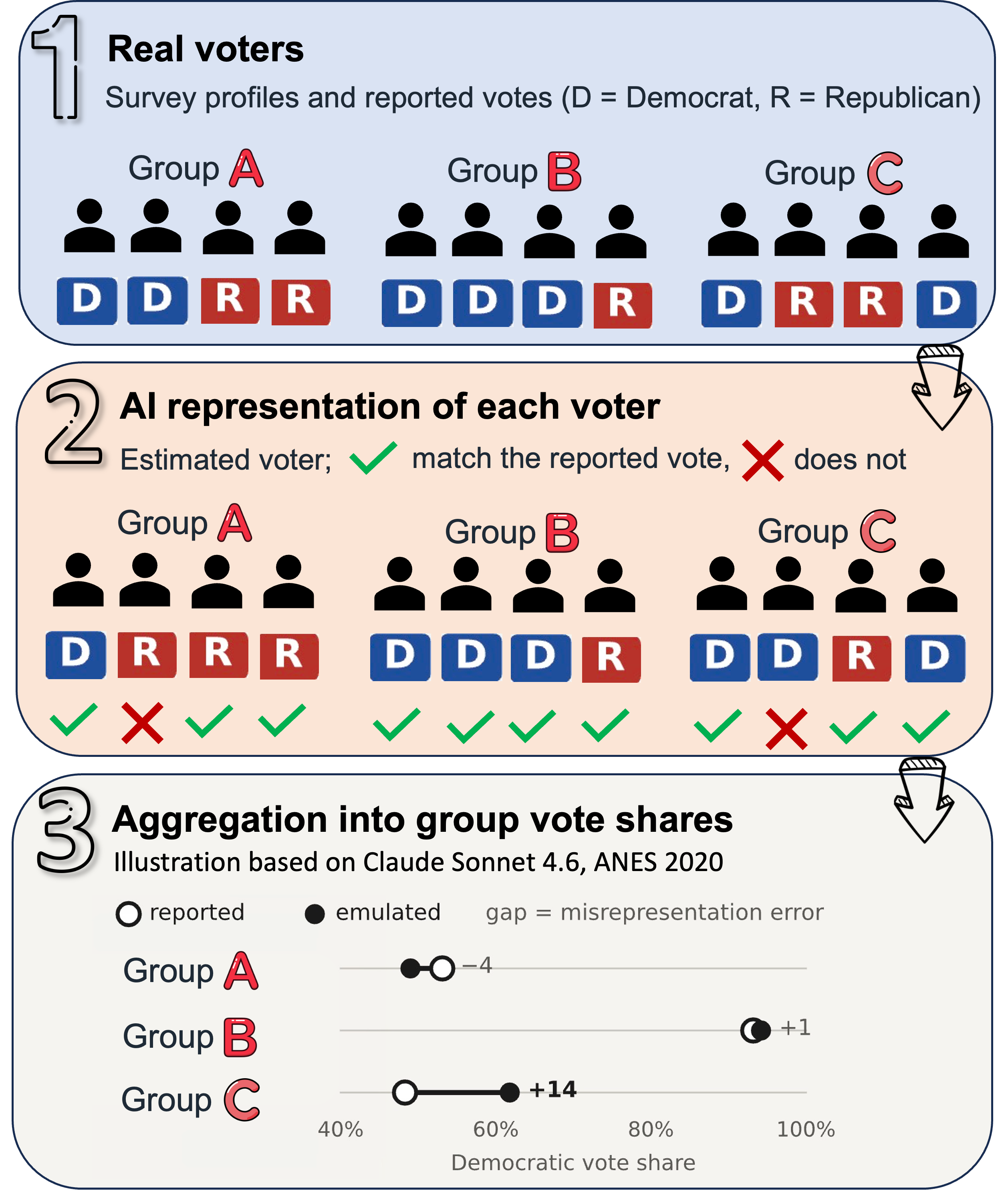}
    \caption{Voter misrepresentation in AI representation. High individual accuracy (85\%) and a small overall error (3 points) can coexist with a large error for one group (Group C, 14 points), reversing its ordering relative to another group.}
    \label{fig:overview}
\end{figure}

Large language models (LLMs) are increasingly used to emulate human
respondents by grounding predictions on demographic, ideological, and
contextual profiles \cite{argyle2023out,zhang2024electionsim,
li2025political,majumdar2026generative,ziems2024can}. In electoral settings, an LLM can
predict a voting choice for each real voter profile and aggregate these
predictions into a synthetic electorate. We refer to this process as
\emph{AI representation of voters}. 
It is used to pilot and complement pre-election polling and forecasting \cite{cerina2025democratic,markaki2026predictive}, to estimate the preferences of abstaining or hard-to-reach voters who are missing from surveys \cite{majumdar2026generative,gholamnia2026enhancing}, to build synthetic populations for policy and social-science research \cite{argyle2023out,zhang2024electionsim}, and to support deliberation and collective decision-making in digital democracy \cite{majumdar2026generative,rymon2026people}. In each of these uses, AI representation is valuable only if the resulting electorate reproduces aggregate vote shares and individual voters and demographic groups are represented faithfully and as intended.
We distinguish \emph{individual fidelity}, where an emulated vote matches
the corresponding human respondent's reported vote, from \emph{aggregate
accuracy}, where the simulated electorate reproduces empirical vote-share
distributions. These attributes do not align: aggregate agreement can
conceal errors that are systematically concentrated within particular voter
groups. 
Prior work on LLM-based opinion simulation shows that model responses to
survey and political questions shift with the demographic profile the model
is conditioned on, and often align better with some groups' opinions than
others \cite{argyle2023out,sun2024random,lee2024can,santurkar2023whose,bang2024measuring}.
Studies that emulate voters specifically report a second problem: the
predicted vote totals can miss the real election result, and the predicted
differences between groups can be larger than those observed in real
voters \cite{kreutner2026persona,potter2024hidden,von2024united,yu2024large}.
What remains open is whether the error in predicting individual votes is
itself unevenly distributed across voter groups, and whether this unevenness
distorts the group vote shares of the emulated electorate.

\noindent\textbf{RQ1.}~\emph{
Is the vote-share error of LLM-emulated electorates equal across demographic groups?
}
We refer to such unequal predictive performance as \emph{representation disparity}: a model may attain reasonable overall accuracy while reconstructing the preferences of some demographic groups less faithfully than others. When errors within a group are directional rather than random, for example when the model predicts too many Democratic votes for one group, they shift that group's emulated vote share away from its reported share, even when overall accuracy appears acceptable (Figure~\ref{fig:overview}). We call such a group \emph{misrepresented}. This matters because many uses of emulated electorates rely on group-level estimates, such as a group's vote share in a poll, and because group errors that do not cancel out shift the overall vote share and can change the emulated winner in close elections. A natural response is repeated sampling, combining multiple generations by majority voting or vote-share averaging \cite{yu2024large,majumdar2026generative}. This reduces random variation but cannot correct errors the model makes consistently: a mistake repeated in every run survives the majority vote. Correcting such errors requires changing how the model represents each voter, which motivates our second question.

\noindent\textbf{RQ2.}~\emph{
Does steering individual emulated voters before aggregation reduce group-level misrepresentation without lowering individual accuracy?
}
To answer RQ2, we compare three levels of control: the base model, a
prompt-based instruction, and a decomposed prompt that separates
vote-relevant reasoning from group representation; for the open-source
model, we also add activation-level control. We evaluate them on every
eligible respondent of ANES 2012, 2016, and 2020. 

Our contributions map onto the two research questions:
\begin{itemize}
\item \emph{Measuring misrepresentation (RQ1).} We define misrepresentation as the gap between a group's emulated and reported vote shares and show, across three elections and two models, that it reaches 10 points or more while individual accuracy stays near 85\%.
\item \emph{Which voter information is fair to use (RQ1).} 
We show that removing gender from the voter profile reduces the largest gender-group error of the voter choices emulated by Claude by 42\% and of the voter choices emulated by Llama by 22\%, and removing education reduces the largest education-group error of the voter choices emulated by Claude by 64\% (p < 0.01), whereas removing self-reported race increases the mean group error by over 40\% (p < 0.05), all without lowering accuracy.
\item \emph{Testing alignment controls (RQ2).} We show that repeated queries, aggregation, fairness instructions, and activation-level control do not reduce misrepresentation, and that a fairness instruction can double the error for a group.
\end{itemize}

\section{Background}
In this section, we reviews how large language models have been used to emulate survey respondents and voters, and methods for measuring and reducing unfairness in LLM-based decision making. We combine the two by measuring fairness within voter emulation and testing whether such methods reduce it.

\subsection{Generative AI Representation of Voters}
Early work on LLM-based human emulation established the possibility of constructing synthetic samples of real survey respondents by prompting language models on real respondent profiles \cite{qi2025representation,aher2023using,horton2023large,dillion2023can}. Argyle et al. introduce “silicon samples” by prompting GPT-3 on socio-demographic backstories drawn from human survey respondents and use the term algorithmic fidelity to show the extent to which model-generated responses reproduce patterns observed in human samples \cite{argyle2023out}. Their results establish demographic prompting as a basis for LLM-based social simulation, but focus primarily on correspondence between synthetic and human response distributions rather than disparities in predictive error across voter groups \cite{linegar2023large}.
Subsequent work has identified limitations of this approach \cite{shrestha2024beyond,zhao2025large}. 
Bisbee et al. condition ChatGPT on profiles derived from ANES 2016 and 2020 respondents and find that synthetic responses can reproduce average public-opinion patterns while exhibiting substantially reduced variation and altered relationships between respondent characteristics and outcomes. 
They further show subgroup differences by race and partisanship, illustrating that aggregate agreement can conceal distortions within the synthetic sample \cite{bisbee2024synthetic}. We build on this observation by measuring such distortions in vote choice, comparing each group's emulated vote share with its reported vote share, and testing whether alignment interventions reduce them.
More recent work focuses directly on electoral behaviour \cite{ma2025algorithmic,chen2024simulating,ashokkumar2026large}. 

Von der Heyde et al. use individual-level voter information to predict voting behaviour in the 2024 European Parliament elections and find that prediction accuracy varies substantially across national and linguistic contexts \cite{von2024united}. Yu et al. develop a large-scale voter-simulation framework using demographic, ideological, and temporal information and evaluate multiple LLMs across U.S. election settings. Their analysis shows that simulated voting patterns can overstate demographic differences involving gender, ethnicity, age, and education, even when aggregate electoral accuracy improves \cite{yu2024large}.
Majumdar et al.\ show that fair collective choice can be resilient to LLM biases and inconsistencies: under several aggregation methods and across proprietary and open-source models, AI-generated choices agree with human choices more consistently at the collective level than at the individual
level \cite{majumdar2026generative}.
However, that work focuses on the robustness of collective-choice mechanisms rather than on whether demographic groups experience systematically different representation errors before aggregation \cite{park2024llm,lin2024large}. Our findings complement this result, which aggregation absorbs random individual inconsistencies, but errors that are
systematic within a demographic group persist through aggregation and
distort that group's emulated vote share.

\subsection{Fairness Analysis in LLM Decision Making}

A parallel literature examines fairness in LLM-supported decision making \cite{hall2026guiding}, and fairness in LLM-based simulation \cite{abeliuk2026auditing} and decision support. Hall et al.\ propose Fairness Reward Models, which improve the fairness of LLM decisions in high-stakes tasks such as recidivism prediction and hiring \cite{hall2026guiding}. The LLM generates several chains of reasoning for the same case, a separate model rates each chain for bias, and the final decision gives more weight to the less biased chains.
Their evaluation explicitly measures disparities using equalized opportunity and equalized odds gaps \cite{hardt2016equality} across protected attributes including race and gender, showing that fairness-oriented inference-time selection can reduce group disparities while maintaining predictive accuracy.
Related work has studied both the measurement of prompt-induced political bias in LLMs \cite{faulborn2025only,rettenberger2025assessing} and 
prompt-level and representation-level interventions for controlling multiple model objectives \cite{c6c06c9474dc4ca2a8280d7ee02129d6,liu2024evaluating,lu2024investigating}. 
These methods define fairness as similar treatment across groups, which suits decisions such as hiring but not voter emulation, where the reported vote shares differ between groups. The fair model must reproduce each group's actual vote share, and an intervention that equalises predictions across groups can increase misrepresentation rather than reduce it \cite{li2025fairsteer,tan2025unmasking,haller2024opiniongpt}. In contrast, our setting combines these two strands, we evaluate fairness within voter emulation by comparing representation error across demographic groups, and then test whether interventions applied before aggregation and at representative process reduce these disparities in the resulting emulated electorate.

\begin{table}[h]
\centering
\caption{Comparison of closely related work across six dimensions of fairness evaluation in AI voter representation.}
\label{tab:related_work_comparison}
\scriptsize
\setlength{\tabcolsep}{2pt}
\begin{tabular*}{\columnwidth}{@{\extracolsep{\fill}}lcccccc@{}}
\toprule
& \shortstack{Per-group\\comparison}
& \shortstack{Fairness\\metric}
& \shortstack{Race/\\ethnicity}
& \shortstack{Multiple\\elections}
& \shortstack{Open \&\\closed LLMs}
& \shortstack{Voter\\emulation} \\
\midrule
Argyle~\cite{argyle2023out}          & $\triangle$ & \xmark & \cmark & \xmark & \xmark & \cmark \\
Bisbee~\cite{bisbee2024synthetic}    & $\triangle$ & \xmark & \cmark & \cmark & \xmark & \cmark \\
von der Heyde~\cite{von2024united}   & \cmark & \xmark & $\triangle$ & \xmark & \xmark & \cmark \\
Yu~\cite{yu2024large}                & $\triangle$ & $\triangle$ & \cmark & \cmark & \cmark & \cmark \\
Majumdar~\cite{majumdar2026generative} & $\triangle$ & $\triangle$ & $\triangle$ & \cmark & \cmark & \cmark \\
Hall~\cite{hall2026guiding}          & \cmark & \cmark & \cmark & -- & $\triangle$ & \xmark \\
\midrule
\textbf{Ours}                        & \cmark & \cmark & \cmark & \cmark & \cmark & \cmark \\
\bottomrule
\end{tabular*}
\begin{minipage}{\columnwidth}
\footnotesize
\textit{Notes:}
\cmark{} = explicit coverage,
$\triangle$ = partial or indirect coverage,
\xmark{} = not addressed,
-- = not applicable.
\end{minipage}
\end{table}
To position our work relative to these two strands of literature, Table~\ref{tab:related_work_comparison} compares closely related studies along six dimensions relevant to our setting whether they: (i) compare predictive performance across groups, (ii) employ a formal fairness metric, (iii) evaluate race or ethnicity, (iv) across multiple election cycles, (v) compare closed-source model and open-source models, and (vi) study voter emulation directly. Prior work addresses several of these dimensions individually, but none of the studies considered here combines all six within a single evaluation framework.

In this table, A partial
mark ($\triangle$) indicates that a study covers the dimension indirectly, for example by reporting subgroup differences without a formal fairness metric. A cross (\xmark) means the dimension applies to the study but is not addressed, whereas a dash (--) means it does not apply, such as Hall et al.~\cite{hall2026guiding} study fairness in general decision tasks rather than elections, so election cycles are not part of their setting. Prior work covers several of these dimensions individually, but none combines
all six, which is needed to test whether demographic misrepresentation in voter emulation is consistent across models and elections and whether fairness interventions reduce it.

\section{Problem Formulation}
\label{sec:formulation}
We emulate an electorate in four steps. (i) Each survey respondent is
described by a voter profile and their reported vote. (ii) Respondents are grouped by a demographic attribute (e.g. race). (iii) The LLM predicts one vote for each profile. Because predictions vary between runs, the model can be queried $K$ times per voter, and all $K$ predictions are then combined into a single emulated vote by majority, which keeps one vote per voter because each voter casts one ballot. (iv) The emulated votes of all voters in a group are aggregated into a group vote results and compared with the group's reported vote results.

\subsection{Voter Profiles and Voter Emulation}
\label{sec:individual_emulation}

Let $\mathcal{D} = \{(x_j, a_j, y_j)\}_{j=1}^{N}$ denote a set of $N$ voter profiles, where $x_j$ contains individual-level profile information such as ideology, age, education, and political interest, $a_j \in \mathcal{R}$ is the respondent's group attribute (here, self-reported race or ethnicity or other grouping criteria can be used), and $y_j \in \mathcal{Y}$ is the respondent's reported vote. Each vote takes one of two values, $\mathcal{Y} = \{\text{Democrat}, \text{Republican}\}$; respondents who reported a third-party vote or did not vote are excluded. 
Given a voter profile $(x_j, a_j)$, the LLM generates a voting \emph{choice} $\hat{y}_j \in \mathcal{Y}$ for that voter; we call this process \emph{AI choice generation}. The information included in the profile, such as race, gender, or education, is the input that shapes each generated choice.
\emph{Individual fidelity} refers to agreement between this emulated vote and the vote reported by the same respondent; the prediction for respondent $j$ is correct when $\hat{y}_j = y_j$. We measure individual fidelity by individual accuracy $\mathrm{Acc}_{\mathrm{ind}}$, the share of respondents
whose emulated vote matches their reported vote:
\begin{equation}
\mathrm{Acc}_{\mathrm{ind}}
=
\frac{1}{N}
\sum_{j=1}^{N}
\mathbf{1}[\hat{y}_j = y_j],
\label{eq:individual_accuracy}
\end{equation}
where $\mathbf{1}[\cdot]$ is the indicator function.

\subsection{Demographic Grouping}
\label{sec:demographic_grouping}


For group-level evaluation, respondents are partitioned by the attribute
$a_j$. Let $\mathcal{R}$ denote the set of groups and
$N_r = |\{j : a_j = r\}|$ the number of respondents in group
$r \in \mathcal{R}$. This partition lets us test whether emulation errors
are distributed unevenly across groups, because each prediction is checked against the respondent's reported vote $y_j$, and each group's emulated vote share is compared with its reported vote share.

We group respondents by the race or ethnicity they self-identify with for three reasons: survey data show that vote choice differs between voters who self-identify with different racial and ethnic groups in U.S. elections \cite{nicolas2024taxonomy}, LLMs show well-documented racial stereotyping, and fair representation of voters who self-identify with a racial or ethnic group is a central concern in electoral law and democratic legitimacy.
The framework itself applies to any grouping, for example, other demographic attributes such as gender, age, education, religion, or region can reveal stereotyping and under-representation of minority opinions and combinations of attributes can reveal intersectional bias.

\subsection{LLM Repeated Vote Generation}
\label{sec:repeated_generation}



LLMs generate text by sampling each word from a probability distribution rather than always choosing the most likely one. This \emph{stochastic decoding}, controlled by parameters such as temperature \cite{holtzman2019curious}, means that the same voter profile can receive different predicted votes in different runs. We therefore query the model $K$ times per profile, obtaining predictions $\hat{y}_{j,1}, \ldots, \hat{y}_{j,K}$, and summarize them by the share of runs that predict a Democratic vote:
\begin{equation}
\hat{p}_j = \frac{1}{K} \sum_{k=1}^{K} \mathbf{1}[\hat{y}_{j,k} = \text{Democrat}].
\label{eq:vote_distribution}
\end{equation}
The emulated vote is the majority prediction, $\hat{y}_j = \text{Democrat}$
if $\hat{p}_j > 0.5$ and Republican otherwise, following self-consistency decoding \cite{wang2022self}; we use odd $K$ to avoid ties. We measure how consistent the model is for respondent $j$ by $c_j = \max(\hat{p}_j, 1 - \hat{p}_j)$, which equals 1 when all runs agree and 0.5 when they are evenly split. Combining runs is done within each respondent and is distinct from aggregation, which combines different respondents into group voting.

\subsection{Group Misrepresentation Metrics}
\label{sec:representation_metrics}

After obtaining one emulated vote $\hat{y}_j$ per respondent, we aggregate votes within each group $r \in \mathcal{R}$. Because votes are binary, a group's vote is fully described by its Share of emulated Democratic choices minus reported share; the Republican share is its complement and has the same error. We compare the share of emulated Democratic choices $\hat{\theta}_r$ with the reported share of Democratic votes $\theta_r$ of the same respondents:
\begin{equation}
\hat{\theta}_r = \frac{1}{N_r} \sum_{j:a_j=r} \mathbf{1}[\hat{y}_j = \text{Dem}],
\qquad
\theta_r = \frac{1}{N_r} \sum_{j:a_j=r} \mathbf{1}[y_j = \text{Dem}].
\label{eq:vote_shares}
\end{equation}
The \emph{misrepresentation error} of group $r$ is the gap between its emulated and reported shares, and our primary fairness metric is the largest such gap across groups:
\begin{equation}
\varepsilon_r = |\hat{\theta}_r - \theta_r|,
\qquad
\varepsilon_{\max} = \max_{r \in \mathcal{R}} \varepsilon_r.
\label{eq:misrep_error}
\end{equation}
We use the absolute rather than relative difference because a gap of a few points has the same electoral weight regardless of a group's baseline share. We also report the mean group error $\bar{\varepsilon}$ and compare both with the population-level error $\varepsilon_{\text{pop}} = |\hat{\theta} - \theta|$ computed over all respondents.

Because all errors are measured in vote-share points against the same respondents' reported votes, they are comparable across groups within a dataset. Small groups yield noisy estimates, so we report $N_r$ and a 95\% confidence interval for every group's signed error $\hat{\theta}_r - \theta_r$, and call a group \emph{misrepresented} when this interval excludes zero. Unlike individual accuracy (Eq.~\eqref{eq:individual_accuracy}), these metrics capture whether errors cancel out or accumulate within a group.

\subsection{From Individual to Group Errors}
\label{sec:transfer}
Improving individual predictions need not improve group vote shares,
because the emulated vote is a majority over runs. A change in a voter's
vote distribution $\hat{p}_j$ has no effect unless it crosses $0.5$, and
then changes the group share by a full $1/N_r$. Aggregating
probabilities instead, $\tilde{\theta}_r = \frac{1}{N_r}\sum_{j:a_j=r}\hat{p}_j$,
removes this discontinuity.

\begin{lemma}
\label{lem:transfer}
Let $\tilde{\theta}_r^{(1)},\dots,\tilde{\theta}_r^{(M)}$ be group shares
obtained by probability aggregation from $M$ predictors, each with
$|\tilde{\theta}_r^{(m)} - \theta_r| \leq \varepsilon$. Then any convex
combination $\sum_m w_m \tilde{\theta}_r^{(m)}$ also has error at most
$\varepsilon$.
\end{lemma}
\begin{proof}
$|\sum_m w_m \tilde{\theta}_r^{(m)} - \theta_r|
= |\sum_m w_m (\tilde{\theta}_r^{(m)} - \theta_r)|
\leq \sum_m w_m \varepsilon = \varepsilon$.
\end{proof}

No such guarantee holds for majority votes, since the mode is not linear
in the predictors. 
\section{Experimental Setup}
\label{sec:setup}
We emulate the electorates of three U.S. presidential elections with two LLMs, one closed-source and one open-source. We evaluate them at three increasing levels of control: (i) the base LLM, (ii) prompt-based control, and (iii) decomposed control, which separates vote-relevant reasoning from group representation; for the open-source model, we add activation-level control. The base level is run several times per voter, with votes combined by majority.

\subsection{Experimental Design}
\label{sec:experimental_design}

We use one closed-source and one open-source model because the two allow different kinds of control. Claude Sonnet 4.6 is a state-of-the-art model available only through an API, so it can be controlled only through its prompt. Llama-3.1-8B-Instruct is a widely used open-weight model whose internal activations can be read and modified, which allows control inside the model \cite{rimsky2024steering}; its size also allows us to run it locally. Using both lets us test whether our findings hold for prompt-based and internal control. Comparing levels of control with the number of runs fixed shows the effect of control alone, and comparing numbers of runs at a fixed level shows the effect of repeated runs alone; for each setting we report both individual accuracy and group misrepresentation error.

\subsection{Datasets and Voter Profiles}
\label{subsec:datasets}

We use the American National Election Studies (ANES)
2012, 2016, and 2020 and include every respondent who reported a
two-party presidential vote and has complete profile data ($n = 3{,}793$,
$1{,}978$, and $5{,}141$; Table~\ref{tab:datasets}). Profiles state race
and ethnicity, gender, age, education, ideology, and political attention.
We group voters by race, gender, age, and education, and exclude groups
with fewer than 30 respondents.

\begin{table}[htbp]
\centering
\caption{Datasets used in the evaluation}
\label{tab:datasets}
\footnotesize
\resizebox{\columnwidth}{!}{%
\renewcommand{\arraystretch}{1.1}
\begin{tabular}{lllll}
\toprule
\textbf{Dataset} & \textbf{Election} & \textbf{Vote task} & \textbf{$N$} & \textbf{Role} \\
\midrule
ANES 2012 & Obama vs.\ Romney &Presidential vote& 3{,}793&Temporal validation \\
ANES 2016 & Clinton vs.\ Trump&Presidential vote & 1{,}978&Temporal validation \\
ANES 2020 & Biden vs.\ Trump&Presidential vote & 5{,}141&Main evaluation \\
\bottomrule
\end{tabular}}
\begin{minipage}{1\columnwidth}
\footnotesize
\emph{Note.} ANES data: \url{https://electionstudies.org/data-center/}. N counts respondents who reported a two-party presidential vote and have complete profile data.
\end{minipage}
\end{table}

\subsection{Models and Levels of Control}


\paragraph{Model} Claude Sonnet 4.6 is available only through an API
and can be controlled only through its prompt; Llama-3.1-8B-Instruct is open-weight, so its internal activations can also be modified
\cite{rimsky2024steering}. Both sample at temperature $\tau = 0.7$ \cite{larsen2025instability,shin2025enhancing}.\footnote{Claude: at most 60 output tokens. Llama: top-$p = 1$, at most 40 output tokens. An answer is read as a vote for the first candidate it names; answers naming no candidate (none for Claude, under 1\% for Llama) are excluded. Prompts and further details are in the supplementary
material.}
\paragraph{Levels of Control} We evaluate each model at three increasing
levels of control (prompts in the supplementary material). (i)
\emph{Base}: the model receives the voter profile and is asked for a
vote. (ii) \emph{Prompt-based control}: a single instruction is added,
asking the model to consider individual electoral factors such as
ideology and political interest, and to ``avoid systematically
under-representing the voting preferences of the voter's racial or ethnic
community'' (verbatim prompt). This instruction assumes that voters who self-identify with the same racial or ethnic group share a voting preference; whether such an instruction helps is part of what we test. (iii) \emph{Decomposed control}: the same two objectives are
placed in separate blocks, one on individual-level signals and one on
community representation, which is to be treated as an independent
consideration that does not override the profile, followed by the final
prediction. Both controls target race, the attribute of RQ1. For Llama,
we add (iv) \emph{activation-level control} by contrastive activation
addition (CAA) \cite{rimsky2024steering}: we take the hidden state of the
last prompt token at layer $\ell = 16$ for 100 voters who self-identify as White and 100 voters who self-identify with another racial or ethnic group (sampled at random from ANES 2020), and use the dominant direction separating the two sets, and add it
with coefficient $\alpha = 1$ during generation. We chose $\alpha$ as the
largest value that keeps the overall Democratic share within 10 points of
the uncontrolled model; larger values collapse almost all predictions to
one candidate. The base level is run once ($K = 1$) and five times
($K = 5$) per voter; the controlled levels are run once and compared with
the first base run on the same voters.

\subsection{Evaluation Metrics}

For each model, dataset, level of control, and number of runs, we report individual accuracy $\mathrm{Acc}_{\mathrm{ind}}$
(Eq.~\eqref{eq:individual_accuracy}), to check whether fairness gains come at the cost of prediction quality; the group errors $\varepsilon_r$ and the worst-group error $\varepsilon_{\max}$
(Eq.~\eqref{eq:misrep_error}), our primary fairness metric, together with the mean group error $\bar{\varepsilon}$ and the population-level error $\varepsilon_{\text{pop}}$. We call a group misrepresented when the 95\% confidence interval of its signed error excludes zero.
\section{Group Misrepresentation (RQ1)}
\label{sec:audit}
This section answers RQ1 using the base level, where no control is
applied, with one run ($K=1$) and five runs per voter ($K=5$). Since the model is not steered at either setting, any difference between them comes from repeated runs alone.

\subsection{Misrepresentation Across Models}
\label{subsec:audit-main}

\begin{figure}[htbp]
\centering
\includegraphics[width=0.9\columnwidth]{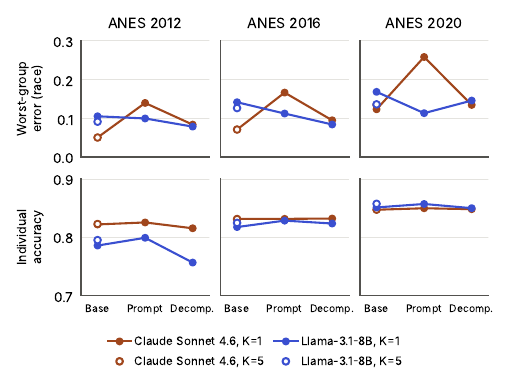}
\caption{Worst-group error over voters who self-identify with a racial or ethnic group (top) and individual accuracy (bottom) at each level of control, with one run per voter (solid) or five runs at the base level (hollow). Accuracy is stable across levels, while race-targeted control does not reduce the worst-group error: for Claude, prompt-based control raises it in every election.}
\label{fig:main_results}
\end{figure}

In ANES 2020, both models reach high accuracy (Claude 0.848, Llama 0.858)
and small population-level errors (3.1 and 1.4 points), yet misrepresent
particular voters who self-identify with a racial or ethnic group (Fig.~\ref{fig:overview}). 
For voters who self-identify as Native American, the voter choices emulated by Claude and the voter choices emulated by Llama place the Democratic share 13.5 and 13.6 points above the reported share, respectively (95\% intervals excluding zero),
while for voters who self-identify as Black, the error of both models stays within 1 point.
The models otherwise differ: 
The voter choices emulated by Claude places the Democratic share of voters who self-identify as White 4.1 points below their reported share, whereas the voter choices emulated by Llama places the Democratic share of voters who self-identify as Asian/NHPI 7.8 points above their reported share.
\begin{figure}[h]
\centering
\includegraphics[width=0.9\columnwidth]{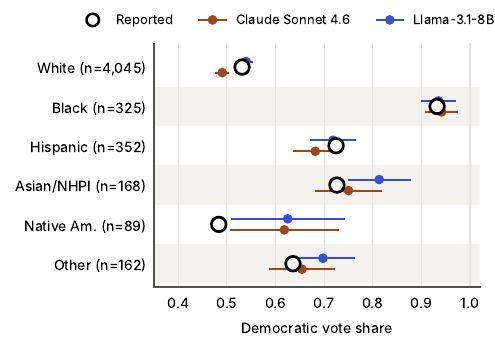}
\caption{
Signed misrepresentation error (share of emulated Democratic choices minus reported share of Democratic votes) for the voter choices emulated by Claude and by Llama, by self-identified racial group, ANES 2020 (base level, K = 5; bars: 95\% intervals). For both models, the share of emulated Democratic choices is furthest from the reported share for voters who self-identify as Native American (about 14 points higher) and within 1 point for voters who self-identify as Black.
}
\label{fig:per_group}
\end{figure}

Figure~\ref{fig:per_group} compares each group's share of emulated Democratic choices with its reported share. 
The largest gap for both models is for voters who self-identify as Native American (n = 89): the voter choices emulated by Claude places their Democratic share 13.5 points above the reported share, and the voter choices emulated by Llama 13.6 points above; as the smallest group, this is also the least precise estimate.
For voters who self-identify as Black or Hispanic, the groups most often examined, the share of emulated Democratic choices of both models stays within 5 points of the reported share.
The two the models then diverge: the voter choices emulated by Claude places the Democratic share of voters who self-identify as White below their reported share, whereas the voter choices emulated by Llama places the Democratic share of voters who self-identify as Asian/NHPI above their reported share.
Because the largest error falls on a group that a comparison limited to voters who self-identify as Black or Hispanic would miss, we evaluate the worst-group error over all groups.

\subsection{Robustness}
\label{sec:robustness}

\begin{figure}[h]
\centering
\includegraphics[width=\columnwidth]{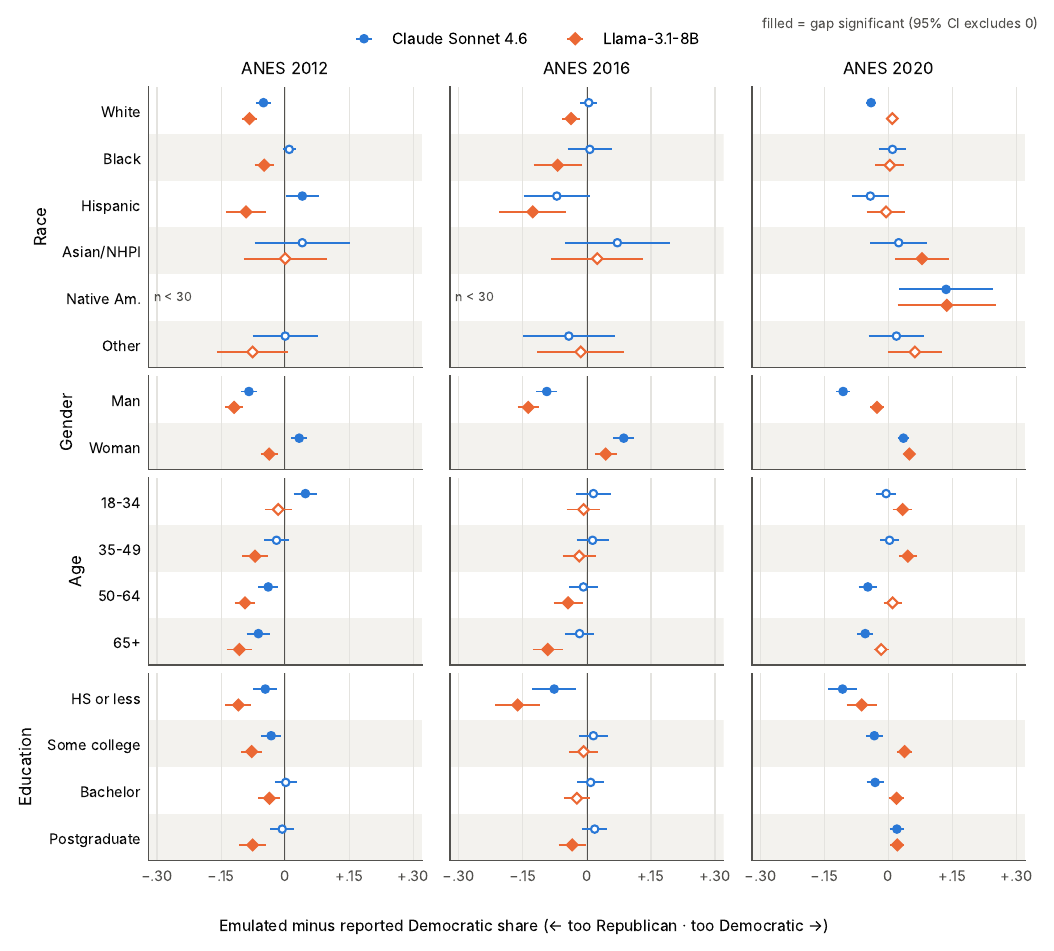}
\caption{Signed misrepresentation error (Share of emulated Democratic choices minus reported share) by group at the base level ($K=5$), with 95\% intervals; filled markers exclude zero. For voters who self-identify as male and voters with at most a high school education, the share of Democratic choices emulated by both models is below the reported share in every election, while the errors for racial groups vary across elections.}
\label{fig:robustness}
\end{figure}

We now test whether this is a stable property of the models or an artifact of three measurement choices: (1) a single run per voter, (2) race as the only grouping, and (3) a single election. We emulate every eligible respondent in ANES 2012, 2016, and 2020 ($n = 3{,}793$, $1{,}978$, and $5{,}141$) at the base level with $K=5$ runs per voter, and report each group's signed error (emulated minus reported share of Democratic votes) with a 95\% confidence interval, for all groups with at least 30 respondents (Fig.~\ref{fig:robustness})\footnote{An answer is read as a vote for the first candidate it names, ignoring neutral pairings such as ``between Biden and Trump''. Answers naming no candidate (none for Claude, under 1\% for Llama) are excluded.}.

\paragraph{Repeated runs.}
If misrepresentation were sampling noise, majority voting over repeated runs would reduce it. It does not. Claude gives the same vote in all five runs for 99.3--99.6\% of respondents, and Llama for 84--90\%. The worst-group error over voters who self-identify with a racial or ethnic group is statistically indistinguishable between one and five runs for both models in every election (paired bootstrap, $p \geq 0.51$). The errors are therefore systematic: the model makes the same mistake for the same voter on every run, and repeated sampling preserves it.

\paragraph{Other demographic attributes.}
Misrepresentation is not specific to race; it is largest and most consistent for gender and education. In all three elections, the share of Democratic choices emulated by Claude is lower than the reported share for voters who self-identify as male (by 8.5, 9.4, and 10.6 percentage points) and higher than the reported share for voters who self-identify as female (by 3.3, 8.6, and 3.5 points), with every 95\% interval excluding zero. As a result, the emulated gender gap is about three times the reported one. For voters with at most a high school education, the share of emulated Democratic choices is likewise lower than the reported share in every election (by 4.7, 7.6, and 10.8 points), widening the reported education gap. Llama errs in the same direction: the share of Democratic choices emulated by Llama is lower than the reported share for voters who self-identify as male (by 12.0, 13.7, and 2.7 points) and for voters with at most a high school education (by 11.0, 16.2, and 6.3 points) in every election, and its emulated gender gap is 2.3–3 times the reported one. 
Errors for racial groups are less stable. For Claude, no racial group is significantly misrepresented in 2016; for Llama, the share of emulated Democratic choices is below the reported share for voters in most self-identified racial groups in 2012 and 2016, in line with its overall shift towards the Republican candidate in those years.
The largest racial error, for voters who self-identify as Native American in 2020, appears in both models (share of emulated Democratic choices 13.5 and 13.6 points higher than the reported share) but rests on the smallest group ($n = 89$). The most consistent errors thus follow one pattern, which the models push gender and education groups further apart than they really are, exaggerating real differences between groups \cite{yu2024large}. These errors remain hidden by aggregate metrics: Claude's population-level error never exceeds 3.1 points, while its errors for large groups reach 10.8 points.

\paragraph{Finer subgroups.}
Splitting voters who self-identify with a racial or ethnic group by gender raises Claude's worst-group error to 0.139, 0.175, and 0.143. Because smaller groups yield noisier estimates, part of this increase is expected even without any additional bias. To separate the two, we shuffle gender labels within each voters who self-identify with a racial or ethnic group, which preserves group sizes and race-level errors, and recompute the worst-group error 1{,}000 times. The observed error exceeds this baseline in 2012 ($p = 0.010$) and 2016 ($p = 0.025$) but not in 2020 ($p = 1.00$); for Llama, only in 2016 ($p = 0.030$; $p = 0.077$ in 2012 and $0.72$ in 2020).

\subsection{Source of Misrepresentation}
\label{sec:source}

\begin{figure}[htbp]
    \centering
  \includegraphics[width=0.9\linewidth]{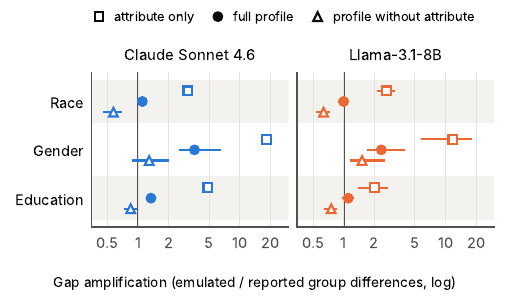}
\caption{Gap amplification on ANES 2020 (1 = real group differences reproduced; log scale; 95\% intervals). Given only the attribute, both models exaggerate group differences several-fold; the full profile still exaggerates gender and education gaps.}
  \label{fig:source}
\end{figure}

\begin{table}[htbp]
\centering
\caption{Group error when one attribute is left out of the profile (ANES 2020). Negative change = fairer.}
\label{tab:without}
\small
\resizebox{\columnwidth}{!}{%
\begin{tabular}{@{}llccccc@{}}
\toprule
& & \multicolumn{2}{c}{Worst-group error} & \multicolumn{2}{c}{Mean group error} & \\
\cmidrule(lr){3-4}\cmidrule(lr){5-6}
Model & Left out & Full $\to$ without & Change & Full $\to$ without & Change & Acc.\ change \\
\midrule
Claude & Race      & .124 $\to$ .185 & $+49\%$ & .043 $\to$ .061 & $+42\%$ & $-0.5$ \\
       & Gender    & .107 $\to$ .062 & $-42\%$ & .071 $\to$ .054 & $-25\%$ & $-0.3$ \\
       & Education & .109 $\to$ .039 & $-64\%$ & .048 $\to$ .021 & $-57\%$ & $+0.2$ \\
\midrule
Llama  & Race      & .136 $\to$ .124 & $-9\%$  & .049 $\to$ .069 & $+41\%$ & $\phantom{+}0.0$ \\
       & Gender    & .049 $\to$ .038 & $-22\%$ & .038 $\to$ .023 & $-39\%$ & $+0.5$ \\
       & Education & .063 $\to$ .072 & $+15\%$ & .035 $\to$ .034 & $-4\%$  & $+1.2$ \\
\bottomrule
\end{tabular}}
\end{table}

We use the profile the model sees for three attributes $A$ (race, gender, and education) on ANES 2020, using the base-level prompt: (i) \emph{attribute only}, a profile stating only $A$, repeated 20 times per group with alternating candidate order; (ii) the \emph{full profile}; and (iii) the full profile \emph{without} $A$.\footnote{Full and without-$A$ prompts use the same respondents and candidate order. Llama uses majority voting over
$K=5$ runs; Claude uses a single run, since its runs agree for 99.6\% of respondents.} We summarize each prompt by its \emph{gap amplification}, the slope of a weighted regression of emulated on reported share of Democratic votes across the groups of $A$ (weights: group size). A value of 1 means the emulation reproduces the real differences between groups, values above 1 mean it exaggerates them, and values below 1 mean it flattens them (Fig.~\ref{fig:source}).

Given only the attribute, Claude returns the same answer in all 20 repetitions for every group, assigning all voters who self-identify with a given group entirely to one candidate: with no other information, the model falls back on a deterministic group-to-party stereotype.
The resulting gaps are 3.1 (race), 18.1 (gender), and 4.8 (education) times the reported ones. Llama follows the same direction less deterministically (2.6, 11.8, and 2.0), and declines to answer half of the gender-only prompts, mostly the prompt stating a male voter (14 of 20).
The full profile dilutes the stereotype but does not remove it.
With the full profile, amplification falls close to 1 for race (Claude:
1.10, 95\% CI [1.00, 1.19]; Llama: 1.00 [0.90, 1.08]) but remains above 1
for gender (3.58 [2.55, 6.36]; 2.34 [1.75, 3.87]) and, for Claude,
education (1.34 [1.21, 1.48]; Llama: 1.11 [0.99, 1.23]). 

Leaving one attribute out of the profile shows which information helps and which harms (Table~\ref{tab:without}; paired bootstrap tests in supplementary material). 
Removing self-reported race increases the group error: without it, the mean group error rises by about 40\% for both Claude and Llama (p = 0.008 and 0.04), and for voters who self-identify as Black, the share of Democratic choices emulated by Claude moves from within 1 point of the reported share to 18.5 points lower. Without race, the models lean on the remaining profile attributes, which do not capture the differences in reported votes between voters who self-identify with different racial or ethnic groups.
Including gender increases the group error: without it, the largest group error falls by 42\% for Claude and 22\% for Llama (both $p < 0.001$), and Claude's emulated gender gap shrinks from 3.6 to 1.3 times the reported one. Including education increases the group error for Claude in the same way: without it, the largest group error falls by 64\% ($p = 0.002$), and the share of emulated Democratic choices for voters without a college degree is no longer 0.4–10.9 points lower than the reported share; for Llama, education makes no significant difference. Accuracy never decreases significantly. Which personal information the model receives therefore decides how fairly groups are represented, at no cost to accuracy.

\subsection*{Answering RQ1}

The vote-share error of emulated electorates is not equal across voters who self-identify with a racial or ethnic group. In ANES 2020, the share of Democratic choices emulated by both models is about 14 points higher than the reported share for voters who self-identify as Native American, while for voters who self-identify as Black it stays within 1 point. Which voters who self-identify with a racial or ethnic group are misrepresented, however, varies across elections and models. The unequal errors also extend beyond race: in every election, the share of Democratic choices emulated by both models is lower than the reported share for voters who self-identify as male and for voters with at most a high school education, exaggerating real differences between groups. They arise because the models weight some demographic labels too heavily: removing gender or education from the profile reduces the errors, whereas removing self-reported race increases them. Giving the model a voter's gender makes the largest gender-group error 29\% (Llama) to 73\% (Claude) larger, and giving it education makes Claude's largest education-group error 2.8 times larger ($p \le 0.002$), without improving accuracy.



\section{Mitigating Misrepresentation (RQ2)}
\label{sec:mitigation}
Misrepresentation is systematic, persists across repeated runs, and is
strongest for gender and education. We test whether prompt-based and
decomposed control reduce it without lowering individual accuracy,
whether their effect extends beyond race, which both target, and why
gains for individual voters do not always carry over to the electorate.

\subsection{Effect of Aggregation}
\label{sec:aggregation}
Before controlling the representation of each voter, we test whether the aggregation step alone can reduce misrepresentation. We compare majority voting over $K=5$ runs ($\hat{\theta}_r$) with probability aggregation ($\tilde{\theta}_r$), which averages each voter's share of Democratic votes $\hat{p}_j$, at the base level on ANES 2012, 2016, and 2020. By Lemma~\ref{lem:transfer}, probability aggregation is never worse than the worst single run; majority voting has no such guarantee.
 
The guarantee holds in practice, but brings little benefit. For Claude,
whose runs agree for over 99\% of voters, the two rules differ by at most 1.5 points for any group. For Llama, whose runs disagree for 10--16\% of voters, the worst-group error changes by at most 1.4 points across race, gender, age, and education, without a consistent direction: probability aggregation is slightly better in some cases and slightly worse in others (e.g., gender in 2020: 4.9 vs.\ 5.4 points). The exaggerated gaps also remain, Llama's emulated gender gap is 2.6, 2.9, and 2.4 times the reported one under probability aggregation, against 2.6, 2.9, and 2.3 under majority voting. Majority voting exceeds the error of the worst single run for only 1 of 46 groups, and the error of the average run for 16 of 46. When a model misjudges a voter, it does so in most of its runs, so
$\hat{p}_j$ is close to 0 or 1 and both rules return the same group
shares. Reducing misrepresentation requires changing how each voter is represented before aggregation.

\subsection{Effect of Control}
\label{sec:control_effect}
 
We emulate every voter once more at the prompt-based and decomposed levels
and compare each with the base level on the same voters, using one run per
voter at every level (Table~\ref{tab:control}, Fig.~\ref{fig:main_results}) and Fig.~\ref{fig:control} shows the resulting error of each group in 2020.
 
\begin{table}[t]
\centering
\caption{Accuracy, population-level error, and worst-group error $\varepsilon_{\max}$ by attribute at each level of control ($K = 1$, same voters). $^{*}$: differs from base ($p < 0.05$, paired bootstrap).}
\label{tab:control}
\small
\resizebox{\columnwidth}{!}{%
\begin{tabular}{@{}lllccccc@{}}
\toprule
& & & & & \multicolumn{3}{c}{$\varepsilon_{\max}$} \\
\cmidrule(l){6-8}
Model & Year & Level & Acc. & Pop. & Race & Gender & Educ. \\
\midrule
Claude & 2012 & Base & .822 & .026 & .051 & .085 & .048 \\
 &  & Prompt & .825 & .043$^{*}$ & .140$^{*}$ & .083 & .091$^{*}$ \\
 &  & Decomposed & .815$^{*}$ & .050$^{*}$ & .084 & .095$^{*}$ & .078$^{*}$ \\
\addlinespace[2pt]
 & 2016 & Base & .832 & .001 & .071 & .094 & .076 \\
 &  & Prompt & .832 & .019 & .167 & .108 & .043 \\
 &  & Decomposed & .832 & .002 & .095 & .084 & .049$^{*}$ \\
\addlinespace[2pt]
 & 2020 & Base & .847 & .031 & .124 & .107 & .109 \\
 &  & Prompt & .850 & .026$^{*}$ & .258$^{*}$ & .095$^{*}$ & .086$^{*}$ \\
 &  & Decomposed & .848 & .047$^{*}$ & .135 & .112 & .101 \\
\midrule
Llama & 2012 & Base & .786 & .083 & .106 & .127 & .122 \\
 &  & Prompt & .799$^{*}$ & .054$^{*}$ & .100 & .094$^{*}$ & .100$^{*}$ \\
 &  & Decomposed & .757$^{*}$ & .019$^{*}$ & .079 & .054$^{*}$ & .069$^{*}$ \\
\addlinespace[2pt]
 & 2016 & Base & .818 & .047 & .142 & .140 & .171 \\
 &  & Prompt & .829$^{*}$ & .034$^{*}$ & .113 & .129 & .153 \\
 &  & Decomposed & .824 & .004$^{*}$ & .085 & .081$^{*}$ & .117$^{*}$ \\
\addlinespace[2pt]
 & 2020 & Base & .851 & .014 & .169 & .052 & .060 \\
 &  & Prompt & .857 & .005 & .114 & .049 & .057 \\
 &  & Decomposed & .850 & .016 & .146 & .057 & .053 \\
\bottomrule
\end{tabular}}
\end{table}

\begin{figure}[htbp]
\centering
\includegraphics[width=0.9\columnwidth]{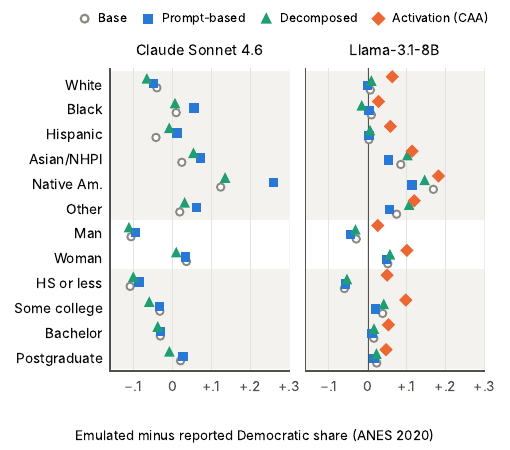}
\caption{Signed misrepresentation error by group on ANES 2020 at each level of control ($K = 1$, same voters). Prompt-based control doubles Claude's error for voters who self-identify as Native American; activation-level control shifts every Llama group toward the Democratic candidate.}
\label{fig:control}
\end{figure}
 
Accuracy changes by at most 1.3 points at either level, except for Llama's
decomposed level in 2012 ($-2.9$). Any change in misrepresentation
therefore comes at almost no cost to individual accuracy, and
accuracy again fails to reveal it.
 
For Claude, the prompt-based instruction moves the shares of emulated Democratic choices for voters in different self-identified racial groups further apart than the reported shares. Where the share of emulated Democratic choices was already higher than the reported share, the gap widens: for voters who self-identify as Native American, from 12.4 to 25.8 points higher in 2020, and for voters who self-identify as Asian/NHPI, from 4.0 to 14.0 points higher in 2012.
The worst racial error roughly doubles (0.051 to 0.140
in 2012, 0.124 to 0.258 in 2020, both $p < 0.05$). The decomposed level
has the same direction but a smaller effect (0.084 and 0.135; not
significant). Neither level reduces the gender errors: for voters who self-identify as male, the share of emulated Democratic choices remains 8–11 points lower than the reported share at every level and in every election.

In 2012 and 2016, where the overall share of Democratic choices emulated by Llama is below the reported share, the decomposed level reduces the population-level error from 8.3 to 1.9 and from 4.7 to 0.4 points, and with it the worst-group error for gender
(0.127 to 0.054; 0.140 to 0.081) and education (0.122 to 0.069; 0.171 to
0.117). The improvement lifts all groups together rather than correcting specific ones: for voters who self-identify as male and voters with at most a high school education, the share of emulated Democratic choices remains 5–12 points lower than the reported share. In
2020, where the base emulation has no overall shift, neither level changes
any worst-group error significantly.
 
\subsection{Gaps under Control}
\label{sec:gaps_control}
 
We next ask whether control reduces the exaggeration of real differences.
For gender, control
reduces amplification for Claude (2020: 3.58 to 3.33 prompt-based and
3.20 decomposed; $p < 0.05$ in all three elections at the decomposed
level), mainly by moving the emulated share for voters who self-identify as female closer to their reported share, but the
emulated gender gap remains 2.6--3.6 times the reported one at every level
and election. For Llama, gender amplification falls significantly only in
2016 (2.96 to 2.56) and remains at least 2.4 everywhere. For race, the
prompt-based level increases Claude's amplification in every election
(2020: 1.10 to 1.27; $p < 0.05$) and the decomposed level in 2016 and 2020
(2020: 1.22), so racial differences become more exaggerated than without
control. Education gaps change little at either
level. No level of control brings the exaggerated gaps close to the
reported ones.
 
\subsection{Prompt vs.\ Activation Control}
\label{sec:prompt_vs_activation}
 
We compare prompt-level and activation-level control on the same model and voters (Llama, ANES 2020, $K = 1$). The decomposed level and activation-level control agree on 91.3\% of voters ($\kappa = 0.82$). This agreement is lower than each method's agreement with the base level (93.3\% and 92.7\%), so it reflects the voters that neither method changes rather than a shared effect. The two methods change the vote of 338 and 367 voters, respectively, but only 133 voters under both.
The direction of the changes also differs. The decomposed level moves
voters in both directions (175 toward the Democratic and 163 toward the
Republican candidate), leaving the overall share unchanged. Activation control moves voters almost only toward the Democratic candidate (315 vs. 52), which raises the population-level error from 1.4 to 6.6 points and raises the share of emulated Democratic choices for every group (Fig.~\ref{fig:control}). It worsens the
worst-group error for gender (0.052 to 0.101; $p < 0.001$) and education
(0.060 to 0.098; $p = 0.08$) without reducing it for race (0.169 to
0.182), while accuracy is unchanged ($-0.6$ points). Its lower education
amplification (1.10 to 0.85) reflects this uniform shift, not a
correction of particular groups. Prompt-level and activation-level control
are therefore distinct interventions, and results obtained with one do
not transfer to the other.

\subsection*{Answering RQ2}
Not with the interventions we tested. Neither aggregation nor control at the level of individual voters reliably reduces misrepresentation.
Probability aggregation changes group errors by at most 1.5 points.
Race-targeted prompts leave the largest errors, for gender and education, in place, and for Claude they exaggerate racial differences further. Where control helps, as for Llama in 2012 and 2016, it does so by removing an overall partisan shift rather than by correcting the groups that are misrepresented. Activation-level control shifts the whole electorate and worsens misrepresentation. Aligning emulated electorates therefore requires controls that adjust how strongly each demographic attribute shapes a voter's representation, and evaluation at the level of groups, since individual accuracy does not reveal whether a control helps or harms.

\section{Conclusion}
AI representation of voters can be accurate for individuals yet unfair to groups. About 85\% of choices are generated correctly, while some groups are misjudged by 10 points or more. 
Removing gender from the voter profile reduces the largest group error of the the models by 22–42\%, and removing education reduces it by up to 64\%, whereas removing self-reported race increases the mean group error by over 40\%, all without lowering accuracy.
Common alignment strategies do not reduce these errors and can double them. Future work will test combinations of voter information, and alignment interventions that adjust how strongly each attribute shapes AI choices, across other elections, countries, and multi-party systems.
\begin{acks}
\parbox{\linewidth}{This work is supported by a UKRI Future Leaders Fellowship (UKRI3495): ``\emph{Smart City Commons of Trustworthy Distributed Intelligence - ARTIO 2.0}''}
\end{acks}

\newpage
\bibliographystyle{ACM-Reference-Format}
\bibliography{sample}

\newpage
\appendix
\section{Theory and Proofs}
\label{app:theory}
 
\subsection{Notation}
\label{app:notation}
 
Table~\ref{tab:notation} summarizes the notation. Let
$\mathcal{D} = \{(x_j, a_j, y_j)\}_{j=1}^{N}$ be the surveyed voters, where
$x_j$ is the voter's profile, $a_j \in \mathcal{G}$ is the voter's group
under a demographic attribute (race, gender, age, or education), and
$y_j \in \{0,1\}$ is the reported vote ($1$ = Democrat). For a group
$r \in \mathcal{G}$, let $J_r = \{j : a_j = r\}$ and $N_r = |J_r|$.
 
\begin{table}[h]
\centering
\caption{Notation.}
\label{tab:notation}
\small
\begin{tabular}{@{}ll@{}}
\toprule
Symbol & Meaning \\
\midrule
$\mathcal{G}$, $r$ & Set of groups under one attribute; a group \\
$J_r$, $N_r$ & Voters in group $r$; their number \\
$y_j$ & Reported vote of voter $j$ ($1$ = Democrat) \\
$\hat{y}_{j,k}$ & Emulated vote of voter $j$ in run $k = 1,\dots,K$ \\
$\hat{p}_j$ & Share of runs voting Democrat, $\frac{1}{K}\sum_k \hat{y}_{j,k}$ \\
$\hat{y}_j$ & Majority vote, $\mathbb{1}[\hat{p}_j > 1/2]$ \\
$\theta_r$ & Reported Democratic share, $\frac{1}{N_r}\sum_{j \in J_r} y_j$ \\
$\hat{\theta}_r$ & Emulated share under majority voting, $\frac{1}{N_r}\sum_{j \in J_r} \hat{y}_j$ \\
$\tilde{\theta}_r$ & Emulated share under probability aggregation, $\frac{1}{N_r}\sum_{j \in J_r} \hat{p}_j$ \\
$\hat{\theta}^{(k)}_r$ & Emulated share from run $k$ alone \\
$\varepsilon_r$ & Misrepresentation error, $|\hat{\theta}_r - \theta_r|$ \\
$\varepsilon_{\max}$, $\bar{\varepsilon}$ & Worst-group and mean group error over $\mathcal{G}$ \\
\bottomrule
\end{tabular}
\end{table}
 
\subsection{Probability Aggregation Preserves Error Bounds}
\label{app:lemma}
 
We state the result for any convex combination of emulations, which covers
both averaging repeated runs of one model and combining several models.
Let $f_1, \dots, f_M$ be emulations that assign each voter a Democratic
probability $f_m(j) \in [0,1]$, with group shares
$\tilde{\theta}_r(f_m) = \frac{1}{N_r}\sum_{j \in J_r} f_m(j)$, and let
$F = \sum_{m=1}^{M} \lambda_m f_m$ with $\lambda_m \ge 0$ and
$\sum_m \lambda_m = 1$.
 
\begin{lemma}[restated]
For every group $r \in \mathcal{G}$,
\[
\bigl|\tilde{\theta}_r(F) - \theta_r\bigr|
\;\le\; \sum_{m=1}^{M} \lambda_m \bigl|\tilde{\theta}_r(f_m) - \theta_r\bigr|
\;\le\; \max_{m} \bigl|\tilde{\theta}_r(f_m) - \theta_r\bigr|.
\]
In particular, if every $f_m$ has error at most $\varepsilon$ for every
group, so does $F$.
\end{lemma}
 
\begin{proof}
The group share is linear in the emulation, so
$\tilde{\theta}_r(F) = \sum_m \lambda_m \tilde{\theta}_r(f_m)$. Since
$\sum_m \lambda_m = 1$,
\[
\tilde{\theta}_r(F) - \theta_r
= \sum_{m=1}^{M} \lambda_m \bigl(\tilde{\theta}_r(f_m) - \theta_r\bigr).
\]
Taking absolute values and applying the triangle inequality gives the first
inequality; the second holds because a convex combination never exceeds
its largest term.
\end{proof}
 
\paragraph{Repeated runs.}
Setting $f_k(j) = \hat{y}_{j,k}$ and $\lambda_k = 1/K$ gives
$F(j) = \hat{p}_j$, so the probability-aggregated share satisfies
\[
\bigl|\tilde{\theta}_r - \theta_r\bigr|
\;\le\; \frac{1}{K}\sum_{k=1}^{K} \bigl|\hat{\theta}^{(k)}_r - \theta_r\bigr|
\;\le\; \max_k \bigl|\hat{\theta}^{(k)}_r - \theta_r\bigr|:
\]
averaging runs is never worse than the average run.
 
\paragraph{Majority voting has no such guarantee.}
Majority voting applies the threshold $\mathbb{1}[\hat{p}_j > 1/2]$ before
averaging, which is not linear. Suppose every voter in group $r$ has
$\hat{p}_j = 0.6$ and $\theta_r = 0.6$. Each run, and probability
aggregation, reproduces the group share in expectation
($\tilde{\theta}_r = 0.6$, $\varepsilon_r = 0$), whereas majority voting
assigns every voter to the Democrat ($\hat{\theta}_r = 1$,
$\varepsilon_r = 0.4$). Majority voting pushes uncertain voters toward the
locally more likely party and can therefore exaggerate differences between
groups. Section~6.1 tests both rules empirically.
 
\subsection{Interpretation}
\label{app:interpretation}
 
The lemma identifies the property an aggregation rule needs for gains at
the level of individual voters to survive aggregation: if each emulation
represents every group within $\varepsilon$, so does their probability
average. It does not imply that aggregation removes misrepresentation, nor
that controlling individual representations is sufficient under majority
voting. When runs share the same directional error, every term on the
right-hand side is large, and the bound offers no improvement.
 
\subsection{Repeated Runs under Independent Noise}
\label{app:noise}
 
The following proposition describes the idealized case in which repeated
runs make independent errors. It is a reference point rather than a
description of our setting.
 
\begin{proposition}
Let $B_{j,k} = \mathbb{1}[\hat{y}_{j,k} \neq y_j]$ indicate an error for
voter $j$ in run $k$. Suppose that, given $x_j$, the $B_{j,1}, \dots,
B_{j,K}$ are independent with $\Pr(B_{j,k} = 1 \mid x_j) \le \delta <
1/2$. Then the majority vote is wrong with probability
\[
\Pr\bigl(\hat{y}_j \neq y_j \mid x_j\bigr)
\;\le\; \Pr\Bigl(\tfrac{1}{K}\textstyle\sum_{k=1}^{K} B_{j,k} \ge \tfrac{1}{2} \,\Big|\, x_j\Bigr)
\;\le\; \exp\!\Bigl(-2K\bigl(\tfrac{1}{2} - \delta\bigr)^{2}\Bigr).
\]
\end{proposition}
 
\begin{proof}
Let $S_K = \sum_{k=1}^{K} B_{j,k}$, so $\mathbb{E}[S_K \mid x_j] \le
K\delta$. The majority vote can be wrong only if at least half of the runs
err, that is, if $S_K \ge K/2$, and
$K/2 - \mathbb{E}[S_K \mid x_j] \ge K(1/2 - \delta) > 0$. Hoeffding's
inequality for independent variables in $[0,1]$ gives
$\Pr\bigl(S_K - \mathbb{E}[S_K \mid x_j] \ge K(1/2 - \delta)\bigr) \le
\exp\bigl(-2K(1/2 - \delta)^2\bigr)$.
\end{proof}
 
Under independent noise, more runs would therefore drive errors toward
zero. Our runs violate this assumption: Claude gives the same vote in all
five runs for 99.3--99.6\% of voters and Llama for 84--90\%, and one and
five runs yield statistically indistinguishable worst-group errors
(Section~5.2). Errors are systematic rather than
independent, so increasing $K$ cannot remove them.

\section{Experiment Design}
\label{app:design}
 
\subsection{Data}
\label{app:data}
 
We use the American National Election Studies (ANES) for 2012, 2016, and
2020. A respondent is included if they reported voting for one of the two
major-party presidential candidates and have complete data on race and
ethnicity, gender, age, education, ideology, and political attention. This
gives $n = 3{,}793$ (2012), $1{,}978$ (2016), and $5{,}141$ (2020)
voters. Group shares are unweighted, so they describe the surveyed
respondents rather than the national electorate.
 
\paragraph{Harmonization.} Each attribute is mapped to the same categories
in all three years (Table~\ref{tab:groups}). The ANES race codes differ
between years (for example, code 3 denotes Asian voters in 2012 and 2016
but Hispanic voters in 2020) and are recoded accordingly. Groups with
fewer than 30 respondents are excluded from all group-level analyses; this
removes Native American voters in 2012 and 2016.
 
\begin{table}[h]
\centering
\caption{Demographic groups used in the analysis.}
\label{tab:groups}
\small
\begin{tabular}{@{}ll@{}}
\toprule
Attribute & Groups \\
\midrule
Race & White, Black, Hispanic, Asian/NHPI, Native American, Other \\
Gender & Man, Woman \\
Age & 18--34, 35--49, 50--64, 65+ \\
Education & High school or less, Some college, Bachelor's, Postgraduate \\
\bottomrule
\end{tabular}
\end{table}
 
\paragraph{Profiles.} Each voter is described by six fields in a fixed
order, for example: \emph{``Race/ethnicity: White; Gender: woman; Age: 40;
Education: a high school diploma or less; Ideology: liberal; Pays
attention to politics: always.''} Ideology uses the seven-point ANES scale
and political attention its five-point scale, both in their verbal labels.
The order in which the two candidates are named in the question is
randomized per voter with a fixed seed and kept identical across all
models and levels of control.
 
\subsection{Models}
\label{app:models}
 
We use one model from each access regime. Claude Sonnet 4.6
(\texttt{claude-sonnet-4-6}) is queried through the Anthropic API with
temperature $0.7$ and a limit of 60 output tokens. Llama-3.1-8B-Instruct
runs locally in bfloat16 on a single NVIDIA A100 GPU with temperature
$0.7$, top-$p = 1$, and a limit of 40 output tokens, which lets its
answers finish after the candidate's name.
 
\paragraph{Answer parsing.} An answer is read as a vote for the first
candidate it names (by full name, surname, or first name), after removing
neutral pairings such as ``between Biden and Trump''. On answers that name
only one candidate, this rule agrees with exact matching in every case.
Answers that name no candidate are excluded: none for Claude and fewer
than 1\% for Llama in every election.
 
\subsection{Levels of Control}
\label{app:levels}
 
Every voter is emulated at three levels of control, using the prompts in
Appendix~\ref{app:prompts}:
\begin{itemize}
\item \emph{Base:} the profile and the vote question only.
\item \emph{Prompt-based:} adds one instruction to consider individual
  electoral factors and to avoid under-representing the voting
  preferences of the voter's racial or ethnic community.
\item \emph{Decomposed:} separates this into a reasoning block on
  individual signals, an independent block on community representation,
  and a final prediction.
\end{itemize}
For Llama, we additionally apply \emph{activation-level} control
(Appendix~\ref{app:caa}) to the base prompt. The number of runs per voter,
$K$, is a separate dimension: base-level emulations use $K = 5$ (run~1
alone gives the $K = 1$ comparison), and controlled emulations use
$K = 1$, so that all levels are compared on the same footing.
 
\subsection{Activation-Level Control}
\label{app:caa}
 
We use Contrastive Activation Addition (CAA). From the ANES 2020 sample, we draw 100 voters who self-identify as White and 100 voters who self-identify as Black, Hispanic, or Asian/NHPI (fixed seed; the two smallest groups are excluded from the contrast), and record the hidden state of the last prompt token at layer 16 of 32 for their base prompts. The steering direction is the top right-singular vector of the activations of the second sample centred on the mean of the first, used because it explains 43.5\%
of their variance (above a pre-set threshold of 40\%; otherwise the mean
difference is used). Its sign is aligned with the mean difference
(cosine $0.95$), and it is normalized to unit length. During generation,
$\alpha$ times this vector is added to the hidden state of the last token
at layer 16.
 
\paragraph{Choice of $\alpha$.} We set $\alpha$ before the main run as the
largest value that keeps the overall Democratic share within 10 points of
the uncontrolled model, on the first 200 voters of ANES 2020
(Table~\ref{tab:alpha}). Larger values collapse almost all predictions to
the Democratic candidate and degrade the fluency of the answers. We use
$\alpha = 1$.
 
\begin{table}[h]
\centering
\caption{Overall Democratic share under activation-level control on 200
ANES 2020 voters (uncontrolled Llama: 0.58; reported: 0.58).}
\label{tab:alpha}
\small
\begin{tabular}{@{}lcccc@{}}
\toprule
$\alpha$ & 1 & 2 & 4 & 8 \\
\midrule
Democratic share & 0.65 & 0.73 & 0.99 & 1.00 \\
\bottomrule
\end{tabular}
\end{table}
 
\subsection{Experiments}
\label{app:experiments}
 
Table~\ref{tab:experiments} lists the experiments behind each part of the
results.
 
\begin{table}[h]
\centering
\caption{Experiments. All use the voters described in
Appendix~\ref{app:data}.}
\label{tab:experiments}
\resizebox{1\columnwidth}{!}{%
\begin{tabular}{@{}llll@{}}
\toprule
Section & Elections & Level of control & Runs \\
\midrule
Sec.5.2 Robustness & 2012, 2016, 2020 & Base & $K = 5$ \\
Sec.5.3 Source & 2020 & Base, three profile variants & $K = 1$ or $5$ \\
Sec.6.1 Aggregation & 2012, 2016, 2020 & Base & $K = 5$ \\
Sec.6.2 Control & 2012, 2016, 2020 & Base, prompt-based, decomposed & $K = 1$ \\
Sec.6.4 Prompt vs.\ activation & 2020 & Decomposed, activation (Llama) & $K = 1$ \\
\bottomrule
\end{tabular}}
\end{table}
 
\paragraph{Source of misrepresentation.} For each attribute $A$ (race,
gender, education), we compare three profiles: $A$ alone (one prompt per
group, repeated 20 times with alternating candidate order), the full
profile, and the full profile without $A$. Llama uses $K = 5$ for the
latter two; Claude uses $K = 1$, as its runs agree for 99.6\% of voters.
 
\subsection{Statistical Analysis}
\label{app:stats}
 
\paragraph{Group errors.} For each group, the 95\% confidence interval of
the signed error $\hat{\theta}_r - \theta_r$ is computed from the paired
per-voter differences $\hat{y}_j - y_j$ as their mean $\pm\, 1.96$
standard errors. An error is called significant when this interval
excludes zero.
 
\paragraph{Comparisons.} Differences between conditions evaluated on the
same voters (one vs.\ five runs; levels of control; aggregation rules;
profile variants) are tested with a paired bootstrap that resamples
voters 1{,}000 times; we report percentile intervals and two-sided
$p$-values.
 
\paragraph{Finer subgroups.} To test whether the larger worst-group error
across race $\times$ gender subgroups exceeds what smaller samples alone
would produce, we shuffle gender within each racial group, which keeps
group sizes and race-level errors fixed, recompute the worst-group error
1{,}000 times, and report the share of shuffles at least as large as the
observed value.
 
\paragraph{Gap amplification.} For an attribute with groups
$r \in \mathcal{G}$, gap amplification is the slope of a weighted
least-squares regression of emulated on reported group shares,
\[
\beta = \frac{\sum_r N_r (\theta_r - \bar{\theta})(\hat{\theta}_r - \bar{\hat{\theta}})}
             {\sum_r N_r (\theta_r - \bar{\theta})^2},
\]
with group-size-weighted means $\bar{\theta}$ and $\bar{\hat{\theta}}$.
$\beta = 1$ means the emulation reproduces the reported differences
between groups, $\beta > 1$ that it exaggerates them, and $\beta < 1$ that
it flattens them. For two groups, $\beta$ is the ratio of the emulated to
the reported gap.

\section{Number of Runs}
\label{app:runs}
 
We test whether more runs per voter reduce misrepresentation by applying
majority voting to the first $K \in \{1, 3, 5\}$ of the five base-level
runs for every voter (Table~\ref{tab:runs}).
 
\begin{table}[h]
\centering
\caption{Accuracy and worst-group error $\varepsilon_{\max}$ by attribute
under majority voting over $K$ runs (base level).}
\label{tab:runs}
\small
\setlength{\tabcolsep}{4pt}
\begin{tabular}{@{}llccccc@{}}
\toprule
& & & & \multicolumn{3}{c}{$\varepsilon_{\max}$} \\
\cmidrule(l){5-7}
Model & Year & $K$ & Acc. & Race & Gender & Educ. \\
\midrule
Claude & 2012 & 1 & .822 & .051 & .085 & .048 \\
       &      & 3 & .822 & .051 & .086 & .048 \\
       &      & 5 & .823 & .051 & .085 & .047 \\
       & 2016 & 1 & .832 & .071 & .094 & .076 \\
       &      & 3 & .832 & .071 & .094 & .076 \\
       &      & 5 & .831 & .071 & .094 & .076 \\
       & 2020 & 1 & .847 & .124 & .107 & .109 \\
       &      & 3 & .847 & .135 & .106 & .108 \\
       &      & 5 & .848 & .135 & .106 & .108 \\
\midrule
Llama  & 2012 & 1 & .786 & .106 & .127 & .122 \\
       &      & 3 & .792 & .098 & .123 & .112 \\
       &      & 5 & .795 & .092 & .120 & .110 \\
       & 2016 & 1 & .818 & .142 & .140 & .171 \\
       &      & 3 & .828 & .135 & .133 & .169 \\
       &      & 5 & .825 & .127 & .137 & .162 \\
       & 2020 & 1 & .851 & .169 & .052 & .060 \\
       &      & 3 & .857 & .159 & .051 & .064 \\
       &      & 5 & .858 & .136 & .049 & .063 \\
\bottomrule
\end{tabular}
\end{table}
 
For Claude, the number of runs makes almost no difference: accuracy and
the worst-group error by gender and education change by at most 0.1
points, because its runs agree for 99.3--99.6\% of voters. The one larger
change, race in 2020 (0.124 to 0.135), comes from the small Native
American group ($n = 89$), where a single voter shifts the group share by
more than one point. For Llama, whose runs disagree more often, accuracy
rises by at most 1.1 points and the worst-group error changes by at most
1.5 points, except for race in 2020 (0.169 to 0.136), which again
concerns the Native American group. The change in the racial worst-group error
between one and five runs is not statistically significant in any
election for either model (paired bootstrap, $p \geq 0.51$;
Section~5.2), and the errors for gender and education,
the largest and most consistent in Section~5.2, stay
essentially unchanged. More runs make the emulation more stable, but they
do not correct a directional error that every run shares.

\section{Prompt Templates}
\label{app:prompts}

This appendix reports the prompt templates used in our experiments. All templates use the same profile, for example: \emph{``Race/ethnicity: White; Gender: woman; Age: 40; Education: a high school diploma or less; Ideology: liberal; Pays attention to politics: always.''} The order of \{Candidate A\} and \{Candidate B\} is randomized per voter and kept fixed across models and levels of control. The model is instructed to answer with one candidate name, and answers are parsed as described in Appendix~\ref{app:models}. Activation-level control uses the base template unchanged (Appendix~\ref{app:caa}).

\begin{table*}[htbp]
\centering
\caption{Prompt templates for the three levels of control. Placeholders in braces are replaced with the voter's profile and the two candidates.}
\label{tab:prompt_templates}
\begin{tabular}{p{0.18\textwidth}p{0.76\textwidth}}
\toprule
Method & Prompt Template \\
\midrule

Unsteered LLM &
\begin{minipage}[t]{\linewidth}
You are simulating a survey respondent in a U.S. election study. Given the following voter profile, predict which candidate this respondent would vote for.

Profile: \{profile\_text\}.

Question: Based on this profile, would this respondent vote for \{Candidate A\} or \{Candidate B\}?

Answer with only one candidate name.
\end{minipage}
\\
\hline

Naive Prompt Steering &
\begin{minipage}[t]{\linewidth}
You are simulating a survey respondent in a U.S. election study. Given the following voter profile, predict which candidate this respondent would vote for.

Profile: \{profile\_text\}.

When making the prediction, consider both individual-level electoral factors such as ideology and political interest, and avoid systematically under-representing the voting preferences of the respondent's racial or ethnic community.

Question: Based on this profile, would this respondent vote for \{Candidate A\} or \{Candidate B\}?

Answer with only one candidate name.
\end{minipage}
\\
\hline

Decomposed Prompt Steering &
\begin{minipage}[t]{\linewidth}
You are simulating a survey respondent in a U.S. election study. Given the following voter profile, predict which candidate this respondent would vote for.

Profile: \{profile\_text\}.

[Shared electoral reasoning block]

Consider task-relevant individual-level signals, including ideology, political interest, age, education, and other
non-sensitive profile attributes.

[Attribute-aware representation block]

Separately consider whether the prediction would systematically under-represent the voting pattern of the respondent's racial or ethnic community. This block should be treated as an independent consideration and should not override the individual-level profile evidence.

[Final prediction]

Based on both blocks, would this respondent vote for \{Candidate A\} or \{Candidate B\}?

Answer with only one candidate name.
\end{minipage}
\\
\bottomrule
\end{tabular}
\end{table*}

\section{Intersectional Results}
\label{app:intersections}
 
Table~\ref{tab:intersections} reports the worst-group error when racial
groups are split by gender, age, or education, at the base level with
$K = 5$. Subgroups with fewer than 30 respondents are excluded, so the
number of subgroups varies across elections. Per-subgroup estimates with
95\% confidence intervals are provided in the supplementary material.
 
\begin{table*}[h]
\centering
\caption{Worst-group error $\varepsilon_{\max}$ for race and its
intersections (base level, $K = 5$). ``Worst subgroup'' gives the subgroup with the largest error, its signed error (positive: share of emulated Democratic choices above the reported share), and its size. ``Sig.'' counts subgroups whose 95\% interval
excludes zero, out of all subgroups with at least 30 respondents.}
\label{tab:intersections}
\small
\setlength{\tabcolsep}{4pt}
\begin{tabular}{@{}lllclrc@{}}
\toprule
Model & Year & Grouping & $\varepsilon_{\max}$ & Worst subgroup & Signed ($n$) & Sig. \\
\midrule
Claude & 2012 & Race               & .051 & White                    & $-.051$ (2437) & 2/5 \\
       &      & Race $\times$ gender    & .139 & White men               & $-.139$ (1205) & 3/9 \\
       &      & Race $\times$ age       & .118 & Other, 35--49           & $-.118$ (34)   & 3/14 \\
       &      & Race $\times$ education & .149 & Hispanic, Bachelor's    & $+.149$ (87)   & 3/14 \\
       & 2016 & Race               & .071 & Asian/NHPI               & $+.071$ (42)   & 0/5 \\
       &      & Race $\times$ gender    & .175 & Hispanic men            & $-.175$ (80)   & 4/8 \\
       &      & Race $\times$ age       & .100 & Hispanic, 18--34        & $-.100$ (40)   & 0/10 \\
       &      & Race $\times$ education & .184 & Hispanic, HS or less    & $-.184$ (38)   & 1/10 \\
       & 2020 & Race               & .135 & Native American          & $+.135$ (89)   & 2/6 \\
       &      & Race $\times$ gender    & .143 & Native American women   & $+.143$ (42)   & 6/12 \\
       &      & Race $\times$ age       & .154 & Other, 50--64           & $+.154$ (39)   & 5/18 \\
       &      & Race $\times$ education & .203 & Hispanic, HS or less    & $-.203$ (74)   & 4/18 \\
\midrule
Llama  & 2012 & Race               & .092 & Hispanic                 & $-.092$ (524)  & 3/5 \\
       &      & Race $\times$ gender    & .154 & Hispanic men            & $-.154$ (266)  & 5/9 \\
       &      & Race $\times$ age       & .152 & Other, 35--49           & $-.152$ (33)   & 8/14 \\
       &      & Race $\times$ education & .216 & Hispanic, HS or less    & $-.216$ (208)  & 7/14 \\
       & 2016 & Race               & .127 & Hispanic                 & $-.127$ (142)  & 3/5 \\
       &      & Race $\times$ gender    & .212 & Hispanic men            & $-.212$ (80)   & 4/8 \\
       &      & Race $\times$ age       & .111 & Black, 35--49           & $-.111$ (45)   & 2/10 \\
       &      & Race $\times$ education & .212 & Black, HS or less       & $-.212$ (33)   & 3/10 \\
       & 2020 & Race               & .136 & Native American          & $+.136$ (88)   & 2/6 \\
       &      & Race $\times$ gender    & .170 & Native American men     & $+.170$ (47)   & 5/12 \\
       &      & Race $\times$ age       & .231 & Other, 50--64           & $+.231$ (39)   & 5/18 \\
       &      & Race $\times$ education & .164 & Hispanic, HS or less    & $-.164$ (73)   & 5/18 \\
\bottomrule
\end{tabular}
\end{table*}
 
Splitting racial groups always raises the worst-group error, from
0.05--0.14 for race alone to 0.10--0.23 for its intersections. Much of
this increase reflects smaller subgroups: several worst subgroups have
fewer than 50 respondents, and for race $\times$ gender, a shuffle test
attributes the increase to sampling noise in some elections
(Section~5.2). One pattern, however, recurs across elections and models and includes some of the larger subgroups (up to 266 voters): for voters who self-identify as Hispanic and male, and for voters who self-identify as Hispanic with at most a high school education, the share of emulated Democratic choices is 15–22 points lower than the reported share. This matches the attributes with the most consistent errors in Section 5.2, where the share of emulated Democratic choices for voters who self-identify as male and for voters without a college degree is below the reported share.

\section{Statistical Tests for the Profile Variants}
\label{app:without_tests}
 
Table~\ref{tab:without_tests} reports, for each attribute left out of the
profile, the change in worst-group error, mean group error, and accuracy
relative to the full profile on ANES 2020, with 95\% confidence intervals
and two-sided $p$-values from a paired bootstrap that resamples voters
1{,}000 times. Claude uses one run per voter and Llama majority voting
over five runs.
 
\begin{table*}[htbp]
\centering
\caption{Change when one attribute is left out of the voter profile (ANES
2020). Differences are ``without'' minus ``full''; negative error
differences mean fairer representation. Bold: $p < 0.05$.}
\label{tab:without_tests}
\small
\begin{tabular}{@{}lllcccrr@{}}
\toprule
Model & Left out & Metric & Full & Without & Difference [95\% CI] & Change & $p$ \\
\midrule
Claude & Race      & Worst-group error & .124 & .185 & $+.061$ [$-.050$, $+.153$] & $+49\%$ & .294 \\
       &           & Mean group error  & .043 & .061 & $\mathbf{+.018}$ [$+.004$, $+.043$] & $\mathbf{+42\%}$ & \textbf{.008} \\
       &           & Accuracy          & .847 & .842 & $-.005$ [$-.011$, $.000$] & & .086 \\
\addlinespace[2pt]
       & Gender    & Worst-group error & .107 & .062 & $\mathbf{-.045}$ [$-.054$, $-.036$] & $\mathbf{-42\%}$ & $\mathbf{<.001}$ \\
       &           & Mean group error  & .071 & .054 & $\mathbf{-.018}$ [$-.031$, $-.004$] & $\mathbf{-25\%}$ & \textbf{.020} \\
       &           & Accuracy          & .847 & .844 & $-.004$ [$-.011$, $+.003$] & & .328 \\
\addlinespace[2pt]
       & Education & Worst-group error & .109 & .039 & $\mathbf{-.070}$ [$-.100$, $-.027$] & $\mathbf{-64\%}$ & \textbf{.002} \\
       &           & Mean group error  & .048 & .021 & $\mathbf{-.027}$ [$-.036$, $-.012$] & $\mathbf{-57\%}$ & \textbf{.002} \\
       &           & Accuracy          & .847 & .849 & $+.002$ [$-.003$, $+.008$] & & .502 \\
\midrule
Llama  & Race      & Worst-group error & .136 & .124 & $-.013$ [$-.054$, $+.058$] & $-9\%$ & .980 \\
       &           & Mean group error  & .049 & .069 & $\mathbf{+.020}$ [$.000$, $+.031$] & $\mathbf{+41\%}$ & \textbf{.040} \\
       &           & Accuracy          & .858 & .858 & $.000$ [$-.004$, $+.006$] & & .818 \\
\addlinespace[2pt]
       & Gender    & Worst-group error & .049 & .038 & $\mathbf{-.011}$ [$-.016$, $-.005$] & $\mathbf{-22\%}$ & $\mathbf{<.001}$ \\
       &           & Mean group error  & .038 & .023 & $-.015$ [$-.026$, $+.001$] & $-39\%$ & .070 \\
       &           & Accuracy          & .858 & .862 & $+.004$ [$.000$, $+.009$] & & .070 \\
\addlinespace[2pt]
       & Education & Worst-group error & .063 & .072 & $+.009$ [$-.050$, $+.067$] & $+15\%$ & .766 \\
       &           & Mean group error  & .035 & .034 & $-.001$ [$-.017$, $+.016$] & $-4\%$ & .998 \\
       &           & Accuracy          & .858 & .869 & $\mathbf{+.011}$ [$+.006$, $+.017$] & & $\mathbf{<.001}$ \\
\bottomrule
\end{tabular}
\end{table*}
 
Leaving out gender reduces the worst-group error significantly for both
models, and leaving out education does so for Claude. Leaving out race
significantly increases the mean group error for both models. Accuracy
never decreases significantly; for Llama, it increases slightly when
education is left out.

\section{Discussion and Limitations}
\label{sec:limitations}
 
\paragraph{What the audit supports}
Misrepresentation is measurable, systematic, and hidden by aggregate
metrics. Across three elections and two models, emulated electorates
reach individual accuracy of 0.80--0.86 and small population-level errors,
yet misjudge particular groups by 10 points or more. The most consistent
errors exaggerate real differences between groups, by gender and
education in particular, and repeated runs or aggregation do not remove
them. Which racial groups are misrepresented, however, varies across
elections, so an audit of one model and election offers limited
assurance for another. Studies that use emulated electorates should
therefore report group-level errors alongside accuracy, per group and
with group sizes, rather than a single summary number.
 
\paragraph{What the mitigation results do not support}
The evidence does not support a claim that the controls we tested reduce
misrepresentation. Race-targeted prompts leave the gender and education
errors in place and, for Claude, exaggerate racial differences further;
where control helps, it removes an overall partisan shift rather than
correcting particular groups; and activation-level control worsens
misrepresentation. Our controls target race, whereas the largest errors concern
gender and education, and the source analysis shows that removing a
demographic label trades one error for another. Reducing
misrepresentation therefore requires controlling how much weight each
label receives, not only whether it is present.
 
\paragraph{Statistical power}
All analyses use every eligible respondent, and group errors are
reported with 95\% confidence intervals. Small groups nevertheless limit
precision: the largest racial error, for voters who self-identify as Native American in 2020, rests on 89 voters, and several intersectional subgroups have fewer
than 50. We therefore report intersectional results in the appendix and
do not claim that they reveal additional misrepresentation beyond what
smaller samples produce.
 
\paragraph{Steering-vector construction}
The activation-level direction contrasts voters who self-identify as White with voters who self-identify with another racial or ethnic group, and is added to every voter regardless of their group. It therefore encodes
demographic identity rather than the accuracy of a voter's
representation, and in practice it shifts the whole electorate toward
the Democratic candidate: even at the strength we use, the Democratic
share rises by 7 points, and larger strengths collapse almost all
predictions to one party (Appendix~\ref{app:caa}). Per-group vectors
applied only to matching voters, or contrasts between correctly and
incorrectly emulated voters within a group, are natural alternatives that
we leave to future work.
 
\paragraph{Prompt confounds}
The prompt-based and decomposed levels differ from the base level in
prompt length, structure, and instruction content at the same time. Any
effect of the decomposed level can therefore not be attributed to its
block structure alone; a fully controlled comparison would require a
single-block prompt of the same length and content. We treat the
decomposed level as one way to implement prompt-level control and scope
our claims accordingly. Activation-level control does not share this
confound, as it uses the base prompt unchanged.
 
\paragraph{Contamination}
Aggregate vote choices by group for the 2012, 2016, and 2020 elections are
widely published and may appear in the models' training data. A model whose share of emulated Democratic choices for a group is close to that group's published aggregate vote share may be recalling the published figure rather than conditioning on the profile it was given. The source analysis (Section~5.3) shows
that both models respond to the full profile and not only to the
demographic label, but it cannot separate conditioning on the profile
from recall of published group statistics.
 
\paragraph{Scope}
We evaluate two models, one per access regime, on one survey series. The
samples include respondents who reported a two-party presidential vote,
and group shares are unweighted, so they describe surveyed voters rather
than the national electorate. Profiles contain six demographic and
political fields and exclude party identification, which would make the
task considerably easier. Whether the findings extend to other countries,
multi-party elections, other survey instruments, or richer profiles
remains to be tested.
\end{document}